\documentclass[a4paper]{article}

\usepackage[margin=1in]{geometry} % full-width

\usepackage{amsmath}
\usepackage{amsthm}
\usepackage{amssymb}
\usepackage{cite}
\usepackage{amsfonts}
\usepackage{enumerate}
\usepackage{graphicx}
\usepackage{braket}
\usepackage{multirow}
\usepackage{xcolor}

\usepackage{tikz}
\usepackage{quantikz}
\usepackage{caption}
\usepackage{subcaption}

\usepackage[utf8]{inputenc}
\usepackage{hyperref}
\theoremstyle{plain}
\newtheorem{theorem}{Theorem}

\newtheorem{lemma}[theorem]{Lemma}

\theoremstyle{definition}

\newtheorem{problem}[theorem]{Problem}

\usepackage{algorithm,algorithmic}

\title{Depth Optimized Ansatz Circuit in
QAOA for Max-Cut}
\author{Ritajit Majumdar$^{1*}$ \and Debasmita Bhoumik$^1$ \and Dhiraj Madan$^2$ \and Dhinakaran Vinayagamurthy$^2$ \and Shesha Raghunathan$^3$ \and Susmita Sur-Kolay$^1$}

\date{
	$^1$Advanced Computing \& Microelectronics Unit, Indian Statistical Institute\\%
	$^2$IBM Research, India\\
	$^3$IBM Systems, India\\~\\
	$^*$majumdar.ritajit@gmail.com
}

\begin{document}
\maketitle
	
\begin{abstract}
		While a Quantum Approximate Optimization Algorithm (QAOA) is intended to provide a quantum advantage in finding approximate solutions to combinatorial optimization problems,  noise in the system is a hurdle in exploiting its full potential. Several error mitigation techniques have been studied to lessen the effect of noise on this algorithm. Recently, Majumdar et al. proposed a Depth First Search (DFS) based method to reduce $n-1$ CNOT gates in the ansatz design of QAOA  for finding Max-Cut in a graph $G = (V,E)$, $|V| = n$. However, this method tends to increase the depth of the circuit, making it more prone to relaxation error. The depth of the circuit is proportional to the height of the DFS tree, which can be $n-1$ in the worst case. In this paper, we propose an $\mathcal{O}(\Delta \cdot n^2)$ greedy heuristic algorithm, where $\Delta$ is the maximum degree of the graph, that finds a spanning tree of lower height, thus reducing the overall depth of the circuit while still retaining the $n-1$ reduction in the number of CNOT gates needed in the ansatz. We numerically show that this algorithm achieves $\simeq 10$ times increase in the probability of success for each iteration of QAOA for Max-Cut. We further show that although the average depth of the circuit produced by this heuristic algorithm still grows linearly with $n$, our algorithm reduces the slope of the linear increase from $\simeq 1$ to $\simeq 0.11$.
		
\end{abstract}

%%%%%%%%%%%%%%%%%%%%%%%%%%%%%%%%%%%%%%%%%
\section{Introduction}
\label{sec:intro}
Quantum Approximate Optimization Algorithm (QAOA) \cite{farhi2014quantum} is a hybrid quantum-classical algorithm \cite{cerezo2020variational}, studied primarily for finding an approximate solution to combinatorial optimization problems. A QAOA is characterized by a Problem Hamiltonian $H_P$ that encodes the combinatorial optimization problem, and a Mixer Hamiltonian $H_M$ that anti-commutes with $H_P$ and whose ground state is easy to prepare. Two parameterized unitaries $U(H_P,\gamma) = exp{(-i\gamma H_P)}$ and $U(H_M,\beta) = exp{(-i\beta H_M)}$ are applied sequentially for $p \geq 1$ times on the initial state $\ket{\psi_0}$. Here $\gamma = \{\gamma_1, \gamma_2, \hdots, \gamma_p\}$ and $\beta = \{\beta_1, \beta_2, \hdots, \beta_p\}$, $\gamma_i, \beta_i \in \mathbb{R}$ $\forall$ $i$, are the parameters. A single epoch of a 
level-$p$ QAOA is represented as in Eq.~(\ref{eq:qaoa}).
\begin{equation}
    \label{eq:qaoa}
    \ket{\psi(\gamma,\beta)} = ( \displaystyle \Pi_{l = 1}^{p} e^{(-i\beta_l H_M)} e^{(-i\gamma_l H_P)}) \ket{\psi_0}
\end{equation}

The parameters are initialized randomly, and after an epoch of iterations, that provides the expectation value $\braket{\psi(\gamma,\beta)|H_P|\psi(\gamma,\beta)}$, the parameters are updated by a classical optimizer. The next epoch uses this new set of parameters and is expected to provide an expectation value that is closer to the optimum solution to the problem.

Farhi et al. first proposed QAOA \cite{farhi2014quantum}, and studied it in the context of finding a maximum cut in a graph, known as the Max-Cut problem. For 3-regular graphs, they showed that a $p = 1$ QAOA achieves an approximation ratio better than random guessing, but lower than the best known classical algorithm \cite{goemans1995improved}. They also argued that the expectation value of the cut produced by the algorithm is a non-decreasing function of $p$. Therefore, QAOA is expected to be a potential candidate for quantum advantage using near-term devices. Many researchers, since then, have studied QAOA in the context of the Max-Cut problem \cite{crooks2018performance, guerreschi2019qaoa, larkin2020evaluation, zhu2020adaptive, yu2021quantum, majumdar2021optimizing}.

A recent experiment from Google \cite{harrigan2021quantum} showed that noise overwhelmed QAOA for Max-Cut in current quantum devices, and the expectation value produced by the algorithm decreases beyond $p = 3$. Error mitigation processes have been studied in the literature that lowers the effect of noise on QAOA, or in general on hybrid quantum-classical algorithms \cite{barron2020measurement, endo2018practical, endo2021hybrid}. Apart from error mitigation techniques, variation in the Mixer Hamiltonian \cite{zhu2020adaptive}, or the Cost Function \cite{larkin2020evaluation, barkoutsos2020improving} have been proposed that either lowers the noise in the circuit, or achieves faster convergence, thus reducing the depth, and hence the effect of decoherence, of the circuit.

In \cite{majumdar2021optimizing}, Majumdar et al. proposed a Depth First Search (DFS) based method that can eliminate $n-1$ CNOT gates in the circuit of QAOA Max-Cut for any graph $G = (V,E)$, where $|V| = n$. Since CNOT gates are one of the primary sources of error in modern quantum devices \cite{ibmquantum}, this procedure significantly reduces the noise in the circuit. As this method imposes an ordering of the edges (discussed in detail in Sec.~\ref{sec:review}), there is an increase in the depth of the circuit. The authors, however, proved that this increase in depth is overshadowed by the reduction in CNOT gates, and the overall circuit has a lower probability of error. Nevertheless, a graph has multiple DFS trees, and the depth of the circuit varies with the height of a DFS tree for the given graph. A circuit with lower depth is naturally preferable with effect of decoherence being lower.

\subsection{Contributions of this article}

It is not a trivial task to find a DFS tree for a given graph that is guaranteed to provide a low depth ansatz circuit. Further, a DFS tree with lower height does not necessarily result in a lower depth circuit (see Sec.~\ref{sec:motivation}). In this paper, we propose a greedy heuristic algorithm to lower the depth of the circuit while still eliminating  $n-1$ CNOT gates. The run-time of our proposed algorithm is $\mathcal{O}(\Delta \cdot n^2)$, where $\Delta$ is the maximum degree of the graph, and it reduces the depth of the circuit by $\simeq 84.8\%$ for graphs with number of vertices $n = 100$. We show our results on Erdos-Renyi Graphs with the probability of edge varying from 0.4 - 0.8, and complete graphs. For graphs with $4 \leq n \leq 12$, we simulate our proposed heuristic with the \textit{ibmq\_manhattan} noise model, and show more than $10$ times increase in the success probability of each iteration of QAOA for Max-Cut. The maximum height of the DFS tree (and hence the depth of the circuit) can be $n-1$, i.e., it increases linearly with the number of vertices with a slope of $\simeq 1$. We show that although the increase in depth of the circuit with the number of vertices from our method is still linear, the slope is reduced to $\simeq 0.11$.

In the rest of the paper,  Sec.~\ref{sec:review}  gives a brief review of QAOA for Max-Cut and the DFS based optimization proposed in \cite{majumdar2021optimizing}. In Sec.~\ref{sec:motivation} and Sec.~\ref{sec:cost_function} we respectively provide the conditions that lead to a low depth circuit, and a corresponding greedy heuristic algorithm to achieve a low depth. Sec.~\ref{sec:result} presents the simulation results of this proposed algorithm, and the concluding remarks appear in Sec.~\ref{sec:conclusion}.

%%%%%%%%%%%%%%%%%%%%%%%%%%%%%%%%%%%%%%%%%
\section{QAOA for Max-Cut and DFS based ansatz optimization}
\label{sec:review}
\subsection{QAOA for Max-Cut}
The traditional QAOA ansatz (the parameterized circuit) is composed of the two parameterized unitaries $U(H_P,\gamma)$ and $U(H_M,\beta)$. Given a graph $G = (V,E)$, with $|V| = n$, the circuit is initialized in the equal superposition of $n$ qubits. This initialization step has a depth of 1 (simultaneous operations of Hadamard gate on all qubits). Similarly, the circuit realization of $U(H_M,\beta)$ is a simultaneous operation of $R_X(\beta)$ on all the qubits.

The operator $U(H_P,\gamma)$ depends on the Problem Hamiltonian. For the Max-Cut problem, $H_P = \displaystyle \sum_{(j,k) \in E} \frac{1}{2}(I-Z_j Z_k)$. The operator $U(H_P,\gamma)$ can, therefore, be represented as in Eq.~(\ref{eq:HP}).
\begin{eqnarray}
\label{eq:HP}
U(H_P,\gamma) &=& \Pi_{(j,k) \in E} U(H_P^{(j,k)}) \nonumber \\
&=& \Pi_{(i,j) \in E} exp(-i \gamma (\frac{I-Z_j Z_k}{2}))
\end{eqnarray}

The circuit realization of the operator $U(H_P^{(j,k)})$ corresponding to an edge $(j,k)$ is shown in Fig. 1. In accordance to the nomenclature used in \cite{majumdar2021optimizing}, we call this circuit a \emph{step}. More precisely, multiple simultaneous such operators can be executed in each step corresponding to disjoint edges. For example, in Fig.~\ref{fig:2reg}, the operators $U(H_P^{(0,1)})$ and $U(H_P^{(2,3)})$ are operated on simultaneously in the same step.

\begin{figure}[H]
\centering
	\begin{quantikz}
		{q_{j}}&&\ctrl{1} & \qw & \ctrl{1} & \qw \\
		{q_{k}}&&\targ{} & \gate{R_z(2\gamma)} & \targ{} & \qw
	\end{quantikz}
	\label{fig:z_jz_k}
	\caption{Circuit realization of the operator $U(H_P^{(j,k)})$}
\end{figure}
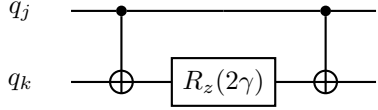

For the rest of the paper, the term \emph{QAOA} will refer to the QAOA for the Max-Cut problem where the graph $G = (V,E)$ is \emph{connected}, \emph{undirected}, and \emph{unweighted}, and $|V| = n$ and $|E| = m$. Every analysis, heuristic and algorithm in this paper will be applicable directly or with minimal changes to weighted graphs or graphs with multiple components. Furthermore, the term \textit{an edge is operated on}, or its equivalent terms, would imply the operation of the circuit of Fig. 1 for that corresponding edge.

For the sake of completeness, in Fig.~\ref{fig:2reg} we show the circuit for a $p = 1$ QAOA corresponding to a 2-regular graph with four vertices.
    
% \begin{figure}[htb]
%     \centering
%     \includegraphics[scale=0.42]{qaoa_circ(1).png}
%     \caption{QAOA circuit for $p=1$  corresponding to a 2-regular graph with four vertices}
%     \label{fig:2reg}
% \end{figure}

\begin{figure}[htb]
    \centering
    \includegraphics[scale=0.42]{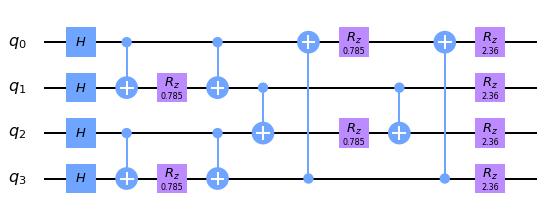}
    \caption{QAOA circuit for $p=1$  corresponding to a 2-regular graph with four vertices}
    \label{fig:2reg}
\end{figure}
    
\subsection{DFS based optimization of the ansatz circuit}
We now briefly discuss the DFS based optimization method that was proposed in \cite{majumdar2021optimizing}. Given a graph $G = (V,E)$, this method finds a DFS tree $T$ starting from a randomly chosen root vertex $r$. The DFS tree is an acyclic subgraph of $G$ containing all the $n$ vertices and $n-1$ edges. The QAOA circuit corresponding to $U(H_P,\gamma)$ can then be partitioned into two portions - the one with edges in $T$, and the other edges. The operators corresponding to the other edges, which are not included in the DFS tree, can be executed in any order once all the operators corresponding to the edges in $T$ are operated on. For each edge $(u,v) \in T$, the operator $U(H_P^{(u,v)})$ is operated on maintaining the following conditions:

\begin{enumerate}[i)]
    \item Starting from the root vertex $r$, if an edge $e_1$ appears earlier than another edge $e_2$ in the DFS tree $T$, then the operator $U(H_P^{e_1})$ must precede $U(H_P^{e_2})$ in the corresponding circuit.
    \item If an edge $e= (u,v)$ is included in $T$ and the vertex $u$ is incident on an edge already in $T$, then $u$ and $v$ must act  as the control and target of the CNOT gate respectively in the operator $U(H_P^{e})$.
\end{enumerate}

Majumdar et al. \cite{majumdar2021optimizing} proved that if these two conditions are satisfied, then the first CNOT gate for the operator in Fig. 1, associated with each edge in $T$, can be removed while still retaining functional equivalence for $p = 1$ of the QAOA circuit. This optimization method does not hold for $p > 1$. Nevertheless, for any level$-p$ QAOA, the first level can be thus optimized.
Therefore, the DFS based method can reduce the CNOT count of the overall QAOA circuit by $n-1$. 

This method mandates a sequential ordering of the tree edges, i.e., an edge in $T$ can be operated only after operators for all of its ancestors have been applied. The maximum height of a DFS tree with $n$ vertices can be $n-1$, thus leading to a significant increase in the depth of the circuit.  Fig.~\ref{fig:ring_of_disagree} shows a 2-regular cycle with $6$ vertices. In the traditional QAOA ansatz, the operator $U(H_P,\gamma)$ can be executed in $2$ steps only. As indicated with the two colors in the graph of Fig.~\ref{fig:ring_of_disagree} (a), the edges assigned the same color can be operated on in the same step. The DFS based method on the other hand (the graph of Fig.~\ref{fig:ring_of_disagree} (b)) requires $6$ steps. It is easy to verify that for a 2-regular cycle, the traditional QAOA ansatz always requires $2$ (for even $n)$ or $3$ (for odd $n$) steps only, whereas the DFS based optimized circuit requires $\mathcal{O}(n)$ steps. For large graphs, it is possible that the depth of the DFS based optimized circuit leads to an execution time greater than the coherence time of the hardware. Even when such a scenario does not occur, increase in depth makes the circuit more susceptible to decoherence. Therefore, although the authors in \cite{majumdar2021optimizing} showed that reduction in the numner of CNOT gates in the ansatz overshadows the increase in depth with respect to probability of error, a pertinent question is whether a DFS tree that reduces the depth of the circuit as well can be found efficiently.  

\begin{figure}[htb]
    \centering
    \begin{subfigure}[b]{0.45\textwidth}
        \centering
        \includegraphics[scale=0.35]{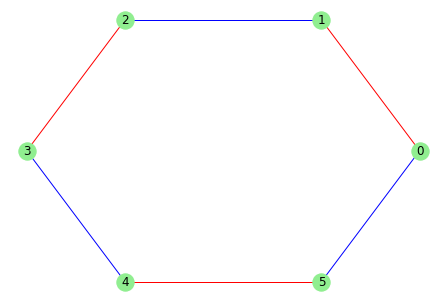}
        \caption{In the traditional QAOA ansatz edges corresponding to the same color can be operated on in the same step; the number of steps needed for this example is 2.}
    \end{subfigure}
    \hfill
    \begin{subfigure}[b]{0.45\textwidth}
        \centering
        \includegraphics[scale=0.35]{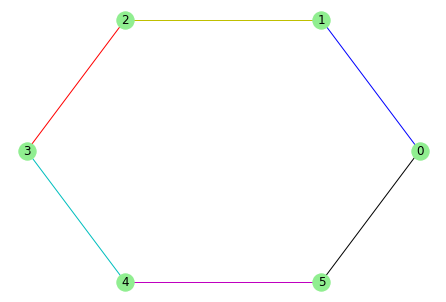}
        \caption{In DFS tree based ansatz optimization --- the number of CNOT gates can be reduced by $5$, but the  number of steps required is $6$.}
    \end{subfigure}
    \caption{Steps of operation for $U(H_P,\gamma)$ for traditional QAOA ansatz and DFS based optimized one.}
    \label{fig:ring_of_disagree}
\end{figure}

Note that this approach to optimizing the number of CNOT gates holds for any rooted spanning tree where the two conditions mentioned above, on the ordering of the tree edges, are met. DFS is merely a method to generate a rooted spanning tree of a graph. Henceforth, instead of specifically finding the DFS tree, we focus on finding a rooted spanning tree for a graph $G$ which serves the above purpose.

%%%%%%%%%%%%%%%%%%%%%%%%%%%%%%%%%%%%%%%%%
\section{Formulation of the Ansatz Optimization Problem}
\label{sec:motivation}
Here we focus only on optimizing for the operator $U(H_P)$ associated with each edge because each of the circuits for initialization and the Mixer Hamiltonian having a depth of one, is the same as for the traditional QAOA circuit \cite{farhi2014quantum}, or for the optimized circuit in \cite{majumdar2021optimizing}. Henceforth, when we mention a circuit, or its depth, we refer to the circuit corresponding to the operator $U(H_P)$ only. 

The maximum height of a DFS tree with $n$ vertices is $n-1$, and so is the depth of the corresponding circuit. It may be claimed that a spanning tree with lower height can lead to a circuit with lower depth. In such a case, a Breadth First Search (BFS) tree may provide a spanning tree with minimum height. However, the two trees shown in Fig.~\ref{fig:tree} (a) and (b) have different heights leading to the same depth of the circuit. In both the figures, the values associated with the edges depict the level at which the operator corresponding to that particular edge can be operated on so that the optimization (i.e. reducing the number of CNOT gates) holds. The circuit corresponding to both of the trees, shown in Fig. 5, is the same. These two trees and their corresponding circuit readily show that simply reducing the height of the tree is not sufficient to obtain a circuit of lower depth. Therefore, finding a BFS tree instead of a DFS tree does not guarantee a circuit with lower depth.

\begin{figure}[htb]
    \centering
    \begin{subfigure}[b]{0.45\textwidth}
        \centering
        \includegraphics[scale=0.35]{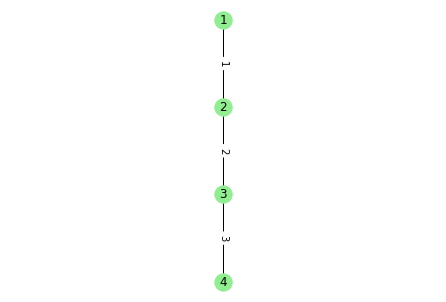}
        \caption{A spanning tree of height 3.}
    \end{subfigure}
    \hfill
    \begin{subfigure}[b]{0.45\textwidth}
        \centering
        \includegraphics[scale=0.35]{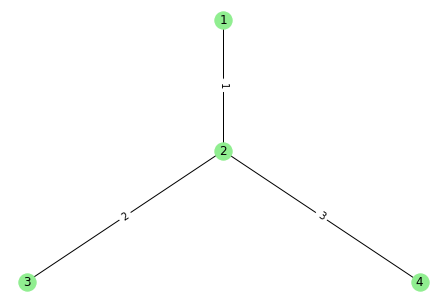}
        \caption{A spanning tree of height 2.}
    \end{subfigure}
    \caption{Two trees with different heights -- the integer label on an edge is the step at which the operator $U(H_P)$ for that edge can be operated on. The maximum value of these labels is the depth of the circuit. The heights of the trees in subfigures (a) and (b) are 3 and 2 respectively. However, both of them lead to the same circuit shown in Fig. 5}
    \label{fig:tree}
\end{figure}

\begin{figure}[htb]
\centering
	\begin{quantikz}
		{q_{1}}&& \qw & \ctrl{1} & \qw & \qw & \qw \\
		{q_{2}}&& \gate{R_z(2\gamma_l)} & \targ{} & \ctrl{1} & \qw & \qw\\
		{q_{3}}&& \gate{R_z(2\gamma_l)} & \qw & \targ{} & \ctrl{1} & \qw\\
		{q_{4}}&& \gate{R_z(2\gamma_l)} & \qw & \qw & \targ{} & \qw
	\end{quantikz}
	\label{fig:circ}
	\caption{The quantum circuit of $U(H_P,\gamma)$ corresponding to both the trees in Fig.~\ref{fig:tree}}
\end{figure}
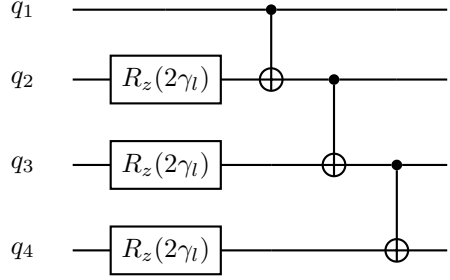

The reason why the tree on the right hand side  with height 2 in Fig.~\ref{fig:tree} cannot lower the depth of the circuit is because the CNOT gates corresponding to two adjacent edges cannot operate at the same step. Therefore, simply reducing the height of the spanning tree is not sufficient to reduce the depth of the circuit. Furthermore, in Fig.~\ref{fig:tree2}, we show two trees with the same height, but having circuits of different depth, as evident from the level number attached with the edges.

\begin{figure}[htb]
    \centering
    \begin{subfigure}[b]{0.45\textwidth}
        \centering
        \includegraphics[scale=0.35]{tree_c.png}
        \caption{An example tree of height 2.}
    \end{subfigure}
    \hfill
    \begin{subfigure}[b]{0.45\textwidth}
        \centering
        \includegraphics[scale=0.35]{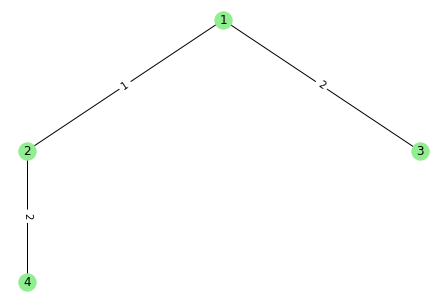}
        \caption{Another example tree of height 2.}
    \end{subfigure}
    \caption{Two trees with same height but the number of steps of the circuit corresponding to the tree in subfigure (a) is 3, while that corresponding to the tree in subfigure (b) is 2.}
    \label{fig:tree2}
\end{figure}

We now define a few terms to easily clarify the requirements for a rooted spanning tree that will lead to a circuit with lower depth.

\begin{enumerate}
    \item \textbf{Branching Factor}: The branching factor of a vertex $v$ is defined as the number of vertices that have been discovered in the rooted spanning tree from $v$.

In other words, the branching factor of a vertex $v$ is one less than the degree of that vertex in the spanning tree except for the root vertex, whose branching factor is equal to its degree. For example, in the tree of Fig.~\ref{fig:tree} (b), starting from the root labelled 1, the branching factor of the root is 1, that of vertex 2 is 2, and that of the leaf vertices are 0.

    \item \textbf{Level}: If a vertex $v$ is discovered in the rooted spanning tree from a vertex $u$, then the level of vertex $v$ = level of vertex $u$ + $1$. The level of the root vertex is $0$.

The definition of level is essentially the same as that for BFS.

    \item \textbf{Delayed Start}: Delayed start is defined as the phenomenon where the vertices $v_1, \hdots v_k$ are discovered in the spanning tree from the same vertex $v$, and belong to the same level, but the edges $(v,v_1), \hdots (v,v_k)$ have to be operated on sequentially. This is because the adjacent edges share a common vertex, and simultaneous CNOT operations are not possible with a common control or target qubit. Therefore, the operation corresponding to the edge $(v,v_i)$ is delayed as long as all the operations corresponding to the edges $(v,v_j), 1\leq j < i$ are not completed. The operator $U(H_P^{(v,v_i)})$ can be operated earliest at the level $level(v) + i$.
\end{enumerate}

We see an example of delayed start in the tree of Fig.~\ref{fig:tree} (b). Although both the leaf vertices in that tree are in the same level, they cannot be operated on simultaneously. Therefore, they must be operated on two disjoint levels. Indeed delayed start is the reason that the depth of the circuit does not reduce directly with the height of the tree. It is obvious that a tree with a higher branching factor can experience more delayed start than a tree with a lower branching factor. On the other hand, the height of the tree increases with decreasing branching factor. Therefore, we need a spanning tree starting from a root vertex $r$ that has as few delayed starts as possible. This rooted spanning tree is neither a DFS tree that has low branching factor, nor a BFS tree that has low height. It is, rather, a tree that has a trade-off between the branching factor (and hence delayed start) and the height.

\subsection{Conjecture: The problem is NP-Complete}

%The requirement of having a rooted spanning tree rules out the usage of Prim's or Kruskal's algorithm for finding a spanning tree for a graph. 

We have seen that simply finding a rooted spanning tree with minimum height is not sufficient to reduce the depth of the corresponding circuit. What we need instead is to have edges that can be executed in parallel. This is similar to the Edge Coloring problem \cite{west2001introduction}. Edges having the same color are disjoint and can be executed in parallel. This method was exploited in \cite{majumdar2021optimizing} as well. However, the problem here is more constrained than the Edge Coloring problem. In a tree, the same color can be used for edges in alternate levels, i.e., it is possible to have edges of same color in level 1, level 3 etc., and in level 2, level 4 etc. For example, in Fig.~\ref{fig:tree} (a), normal Edge Coloring assigns the same color to the edges in levels 1 and  3 since they are disjoint. However, we have already discussed that operators corresponding to these two edges cannot execute in parallel. Therefore, we have an added constraint that an edge cannot be given the color of any of its ancestors. We formally define the problem as follows:

% \begin{problem}
% Given a graph $G = (V,E)$, starting from a root vertex $r$ find a spanning tree $T$ of $G$  whose edges can be colored with the minimum number of colors provided that no edge has the same color as any of its ancestors.
% \end{problem}

\begin{problem}
Given a graph $G = (V,E)$, starting from a root vertex $r$ find a spanning tree $T$ of $G$  whose edges can be colored with the minimum number of colors satisfying the conditions that

\begin{enumerate}[i)]
    \item No two edges incident on a common vertex have same color.
    \item No edge has same color as that of any of its ancestors.
\end{enumerate}
\end{problem}

Optimal Edge Coloring problem is itself an NP-Complete problem, and Problem 1 has additional constraints. In other words, we want that the degrees of the vertices in the spanning tree are not very large in order to avoid delayed start. Edges which are suffering from delayed start due to the rooted spanning tree ordering cannot have the same colors, and will thus increase the required number of colors. In \cite{rahman2005complexities}, the authors showed that finding a degree constrained spanning tree, i.e., a spanning tree where the degree of any vertex is upper bounded by a predefined value, in NP-Complete.

Hence, we \textit{conjecture} that Problem 1 is NP-Complete, and propose a greedy polynomial time algorithm to find a \textit{better} solution instead of the vanilla flavour depth-first search based method.

%%%%%%%%%%%%%%%%%%%%%%%%%%%%%%%%%%%%%%%%%
\section{Proposed Cost Function and Algorithm}
\label{sec:cost_function}
We propose a cost function that respects the following observations:

\begin{enumerate}[i)]
    \item If the branching factors of the vertices are very high, then the corresponding circuit will suffer from \emph{delayed start}, leading to an increase in the depth. On the other hand, if the branching factor of the vertices are very low, then the height of the tree, and hence of the depth of the circuit, will increase.
    
    \item Between two vertices $u$ and $v$, it is better to branch the one at a lower level of the tree so that the edges in that branch may still have some opportunity to be executed in parallel with other edges at a higher level even after \emph{delayed start}. An example of this is shown in Fig.~\ref{fig:tree2} where both the trees have the same height, but the tree of Fig.~\ref{fig:tree2} (b) will lead to a circuit with lower depth since the branching is closer to the root.
    
    \item For graphs with fewer vertices, the branching factor should be low in order to avoid increase in depth due to \emph{delayed start}. However, as the number of vertices increases, a higher branching factor must be allowed to lower the height of the tree.
\end{enumerate}

Respecting all the three criteria stated above, we propose a cost function $C_v$ to be associated with every vertex $v$. Let $n$ be the number of vertices in the graph, $l_v$ and $v_{bf}$ be the level and the current branching factor of the vertex $v$ respectively, and $B$ be the maximum branching factor decided for any vertex in the spanning tree, then
\begin{equation}
    \label{eq:cost_func}
    C_v = (n - l_v) \cdot (B - v_{bf})
\end{equation}

When growing the spanning tree from a root vertex, the edge $(v,w)$ for which the cost function $C_v$ is maximum, is added to the tree. Note here that for a new edge $(v,w)$, the cost function does not depend on the vertex $w$, but rather on the vertex $v$ from which this edge is discovered (the algorithm is provided later on).

The term $(n - l_v)$ is always positive. On the other hand,
\begin{equation*}
  B - v_{bf}
    \begin{cases}
      > 0 & \text{if~} v_{bf} < B\\
      = 0 & \text{if~} v_{bf} = B\\
      < 0 & \text{if~} v_{bf} > B.
    \end{cases}       
\end{equation*}

Our proposed algorithm avoids branching at a vertex for which $v_{bf} \geq B$. In fact, when $v_{bf} = B$, the cost function has a contribution of $0$. Therefore, in our result section, we take $B = f+1$ if we want a maximum branching factor of $f$ in the spanning tree. Furthermore, the term $(n - l_v)$ is higher for the vertices with lower $l_v$. Thus, if for two vertices $u \neq v$, $v_{bf} = u_{bf} < B$, the algorithm  chooses to branch that vertex which has a lower level. This ensures that \emph{delayed start} is closer to the root, so that those branches still have some opportunity for parallel execution with some higher level branches. Furthermore, if $v_{bf} > B$, the product with $(n - l_v)$ leads to significantly low values for low $l_v$. This discourages branching more than $B$ in lower levels of the tree strongly to prevent excessive \emph{delayed start} (like in a BFS tree). In other words, the spanning tree generated by this heuristic cost function is neither a BFS nor a DFS one, but rather an intermediate one. Algorithm~\ref{alg:heuristic}  presents how to generate a spanning tree that has a trade-off between the height of the tree and the branching factor.

\begin{algorithm}
\caption{Cost Function Based Rooted Spanning Tree Generation}
\label{alg:heuristic}
\begin{algorithmic}[1]
\REQUIRE A Graph $G = (V,E)$, $|V| = n$, $|E| = m$; maximum branching factor $B$.
\ENSURE A Rooted Spanning Tree $T$ of the Graph $G$.
\STATE $T = \{\}$.
\STATE $u_{bf} \leftarrow 0$ for all vertex $u$.
\STATE $r \leftarrow$ randomly selected start vertex.
\STATE Visited = $\{r\}$.
\STATE $r_{bf} = r_{bf} + 1$.
\STATE edges\_to\_add = neigh(r).
\WHILE{$|Visited| < n$}
\STATE e = edges\_to\_add[0].
\STATE c = $0$.
\FORALL{$edge = (u,v) \in edges\_to\_add$}
\STATE cost = $(n-l_u) \cdot (B - u_{bf})$.
\IF{$cost > c$}
\STATE c = cost.
\STATE e = edge.
\ENDIF
\ENDFOR
\STATE $T = T \cup \{e\}$.
\STATE Visited = Visited $\cup$ $\{y\}$, where $e = (x,y)$.
\STATE $x_{bf} = x_{bf} + 1$.
\STATE Remove all edges of the form $(*,q)$ from edges\_to\_add.
\FORALL{$edge = (p,q) \in neigh(y)$}
\IF{$q \notin Visited$}
\STATE edges\_to\_add = edges\_to\_add $\cup$ $\{edge\}$.
\ENDIF
\ENDFOR
\ENDWHILE
\end{algorithmic}
\end{algorithm}

\begin{lemma}
Algorithm~\ref{alg:heuristic} finds a rooted spanning tree in $\mathcal{O}(\Delta \cdot n^2)$ time for a graph with $n$ vertices and maximum degree $\Delta$ which satisfies the conditions in Problem 1.
\end{lemma}

\begin{proof}
Let $r$ be the randomly chosen root vertex of the spanning tree. Therefore, the choice of root does not require any computational time. Since $\Delta$ is the maximum degree of the graph, $r$ can have at most $\Delta$ neighbours. Finding the maximum cost function among these neighbours require $\mathcal{O}(\Delta)$ time. Subsequent vertices in the spanning tree can have at most $\Delta - 1$ neighbours since one of its neighbour must be its parent in the spanning tree.  Therefore, the total time requirement in all the steps is
\begin{eqnarray*}
W &\leq& \Delta ~~\text{(to create the spanning tree upto two vertices)}\\
& \leq & \Delta + (\Delta - 1) ~~\text{(to create the spanning tree upto three vertices)}\\
& \leq & 3\Delta - 2 ~~\text{(to create the spanning tree upto four vertices)}\\
& & \vdots
\end{eqnarray*}

Therefore,
\begin{eqnarray*}
W &\leq& \displaystyle \sum_{i=1}^{n-1} (i \cdot \Delta - (i-1))\\
& = & \displaystyle \Delta \cdot \sum_{i=1}^{n-1} i - \displaystyle \sum_{i=1}^{n-1} (i-1)\\
%&=& \Delta \cdot \frac{n(n-1)}{2} - \frac{(n-1)(n-2)}{2}\\
&=& \mathcal{O}(\Delta \cdot n^2)
\end{eqnarray*}
\end{proof}

For sparse graphs, $\Delta = \mathcal{O}(1)$ and for dense graphs $\Delta = \mathcal{O}(n)$. Therefore, the time complexity of the proposed Algorithm~\ref{alg:heuristic} varies between $\mathcal{O}(n^2)$ to $\mathcal{O}(n^3)$ depending on the sparsity of the given graph.

\subsection{An Illustration of Algorithm 1}

We illustrate the DFS method \cite{majumdar2021optimizing} and our proposed method in action on an example graph given in Fig.~\ref{fig:ex_graph}. First, in Fig.~\ref{fig:traditional} we show the traditional $p = 1$ QAOA circuit for this graph. Then, in Fig.~\ref{fig:twotrees} we give two spanning trees of the graph. The spanning tree in Fig.~\ref{fig:twotrees} (a) is generated using the DFS method \cite{majumdar2021optimizing}, whereas the one in Fig.~\ref{fig:twotrees} (b) is generated using Algorithm~\ref{alg:heuristic} with $B = 3$. In Fig. 10 (a) and (b) we show the optimized circuits for the $p = 1$ QAOA of the graph in Fig.~\ref{fig:ex_graph}, where the optimized circuits are generated using the DFS method \cite{majumdar2021optimizing} and Algorithm~\ref{alg:heuristic} respectively. The values of $\gamma$ and $\beta$ are randomly selected.

\begin{figure}[htb]
    \centering
    \includegraphics[scale=0.35]{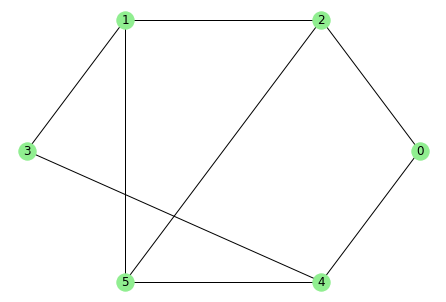}
    \caption{An example graph with 6 vertices}
    \label{fig:ex_graph}
\end{figure}

% \begin{figure}[htb]
%     \centering
%     \includegraphics[scale=0.35]{traditional_qaoa1(1).png}
%     \caption{Traditional $p=1$ QAOA circuit corresponding to $U(H_P,\gamma)$ for the graph in Fig.~\ref{fig:ex_graph}}
%     \label{fig:traditional}
% \end{figure}

\begin{figure}[htb]
    \centering
    \includegraphics[scale=0.35]{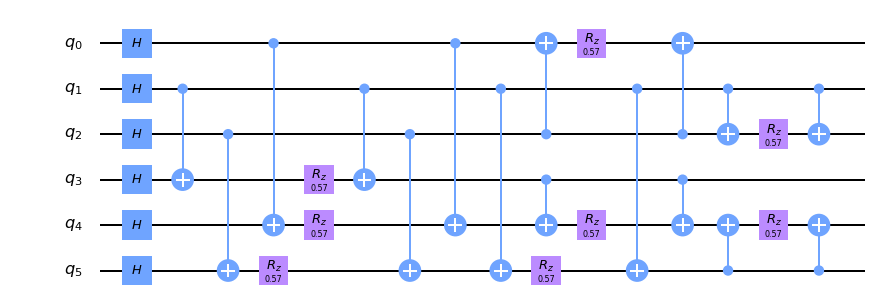}
    \caption{Traditional $p=1$ QAOA circuit corresponding to $U(H_P,\gamma)$ for the graph in Fig.~\ref{fig:ex_graph}}
    \label{fig:traditional}
\end{figure}

\begin{figure}[htb]
    \centering
    \begin{subfigure}[b]{0.45\textwidth}
        \centering
        \includegraphics[scale=0.35]{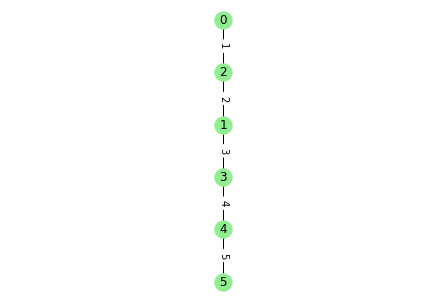}
        \caption{A spanning tree tree of the graph in Fig.~\ref{fig:ex_graph}}
    \end{subfigure}
    \hfill
    \begin{subfigure}[b]{0.45\textwidth}
        \centering
        \includegraphics[scale=0.35]{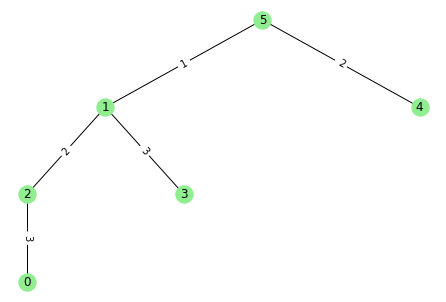}
        \caption{Another spanning tree tree of the graph in Fig.~\ref{fig:ex_graph}}
    \end{subfigure}
    \caption{Two spanning trees of the graph in Fig.~\ref{fig:ex_graph}. The left graph is generated using the DFS method \cite{majumdar2021optimizing}, and the right one is generated using Algorithm~\ref{alg:heuristic} with $B = 3$.}
    \label{fig:twotrees}
\end{figure}

% \begin{figure}[htb]
%     \centering
%      \begin{subfigure}[b]{0.47\textwidth}
%          \centering
%          \includegraphics[scale=0.3]{dfs(1).png}
%          \caption{Optimized circuit using DFS method}
%      \end{subfigure}
%      \hfill
%      \begin{subfigure}[b]{0.47\textwidth}
%          \centering
%          \includegraphics[scale=0.3]{myopt(1).png}
%          \caption{Optimized circuit using proposed Algorithm 1}
%      \end{subfigure}
%      \caption{Optimized $p = 1$ QAOA circuit corresponding to $U(H_P,\gamma)$ for the two spanning trees in Fig.~\ref{fig:twotrees} respectively}
%      \label{fig:circuits}
% \end{figure}

\begin{figure}[htb]
    \centering
     \begin{subfigure}[b]{0.47\textwidth}
         \centering
         \includegraphics[scale=0.3]{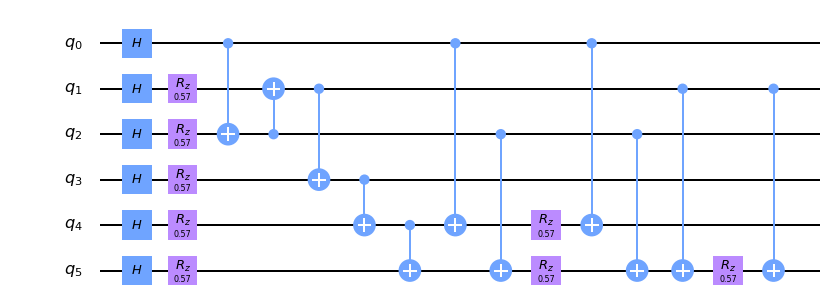}
         \caption{Optimized circuit using DFS method}
     \end{subfigure}
     \hfill
     \begin{subfigure}[b]{0.47\textwidth}
         \centering
         \includegraphics[scale=0.3]{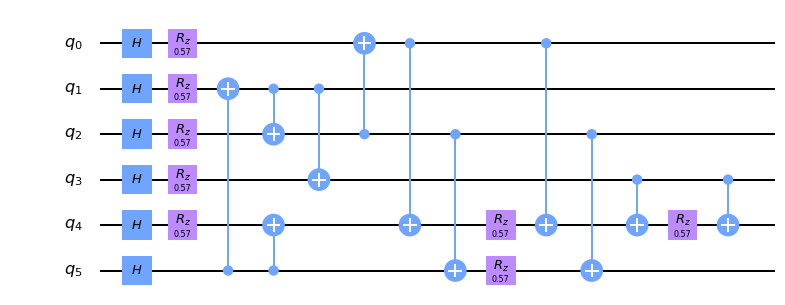}
         \caption{Optimized circuit using proposed Algorithm 1}
     \end{subfigure}
     \caption{Optimized $p = 1$ QAOA circuit corresponding to $U(H_P,\gamma)$ for the two spanning trees in Fig.~\ref{fig:twotrees} respectively}
     \label{fig:circuits}
\end{figure}

The depth of the circuits in Fig.~\ref{fig:traditional}, Fig.~\ref{fig:circuits} (a) and Fig.~\ref{fig:circuits} (b) are 11, 14 and 12 respectively as obtained using the \textit{.depth()} function of Qiskit \cite{Qiskit}. The number of CNOT gates in both the optimized circuits in Fig.~\ref{fig:circuits} are 5 less than that in Fig.~\ref{fig:traditional}. We note that the depth of both the optimized circuits are greater than that of the traditional QAOA. However, the optimized circuit in Fig.~\ref{fig:circuits} (b) can be considered to be superior since it requires 5 CNOT gates fewer than that in Fig.~\ref{fig:traditional}, as well as increases the depth by 1 only. In Sec.~\ref{sec:result}, we show that our proposed Algorithm~\ref{alg:heuristic} can significantly arrest the increase in depth, and in some cases can lead to a lower depth than its traditional circuit. %(Table~\ref{tab:comparison}).

%%%%%%%%%%%%%%%%%%%%%%%%%%%%%%%%%%%%%%%%%
\section{Results of simulation}
\label{sec:result}

\subsection{Reduction in the depth of the circuit}
The entire circuit of $U(H_P,\gamma)$ can be divided into two disjoint parts - one corresponding to the edges in the spanning tree, followed by the other edges in the input graph. Our algorithm can reduce the depth of the circuit corresponding to the spanning tree only. The circuit for the unoptimized edges remains the same as in \cite{majumdar2021optimizing}. Furthermore, the initialization, and the Mixer Hamiltonian is the same for the traditional QAOA circuit \cite{farhi2014quantum} or the circuit proposed in \cite{majumdar2021optimizing}, and our proposed optimized circuit. Therefore, here we compare the depth of the circuit corresponding to the spanning tree only.

When executing a circuit in a hardware, the graph has to be mapped to the underlying hardware connectivity graph. This process is called transpilation. All the results in this section are generated after transpiling the original circuit in the \emph{ibmq\_manhattan} connectivity graph using the \textit{transpilation} procedure of qiskit \cite{Qiskit} with \textit{optimization\_level = 3}.

In the worst case, the height of the DFS tree, and hence the depth of the corresponding circuit, can be as large as $n-1$. In Fig.~\ref{fig:depth} (a)-(d) we show the reduction in depth of of the circuit of the spanning tree by our proposed algorithm compared to the worst case depth of the circuit corresponding to the maximum height of the DFS tree. Fig.~\ref{fig:depth} (a)-(d) show the reduction in depth for Erdos-Renyi graphs with $p_{edge}$, the probability of an edge, varying from 0.4  to 0.8, and complete graphs. For each type of graph, we vary $n$, the number of vertices from 20 to 100, and the value of the depth corresponding to each $n$ is an average over $80$ graph instances. The graph instances are same for all the values of $B$. The graphs in Fig.~\ref{fig:depth} and later in Fig.~\ref{fig:success} are averaged over all the possible $n$ spanning trees generated by selecting each of the $n$ vertices once as the root.

For all the types of graphs considered, we observe higher reduction in the depth for a higher value of $B$ as the number of vertices increases. This is at par with our reasoning earlier, that for larger graphs, it is better to allow higher values of $B$. We have shown results for $B = 3, 6$ and $10$ only. For $B = 10$ and $n = 100$, the reduction in the depth is $\simeq 84.8\%$. It is evident from the graphs that the depth decreases with increasing $B$ as the number of vertices increases. So for larger graphs, one should opt for even higher values of $B$.

\begin{figure}[htb]
     \centering
     \begin{subfigure}[b]{0.48\textwidth}
         \centering
         \includegraphics[scale=0.5]{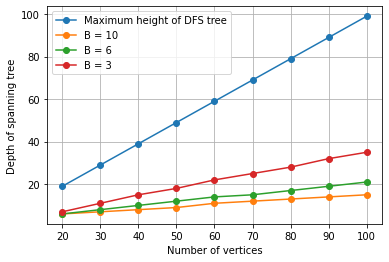}
         \caption{Erdos-Renyi Graphs with $p_{edge} = 0.4$}
     \end{subfigure}
     \hfill
     \begin{subfigure}[b]{0.48\textwidth}
         \centering
         \includegraphics[scale=0.5]{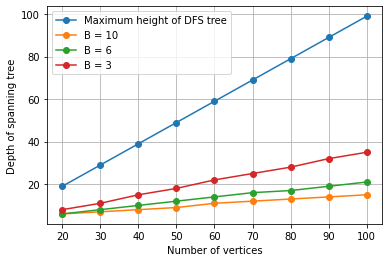}
         \caption{Erdos-Renyi Graphs  with $p_{edge} =  0.6$}
     \end{subfigure}
     \newline
     \begin{subfigure}[b]{0.47\textwidth}
         \centering
         \includegraphics[scale=0.5]{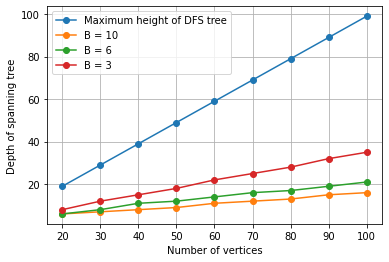}
         \caption{Erdos-Renyi Graphs  with $p_{edge} = 0.8$}
     \end{subfigure}
     \hfill
     \begin{subfigure}[b]{0.47\textwidth}
         \centering
         \includegraphics[scale=0.5]{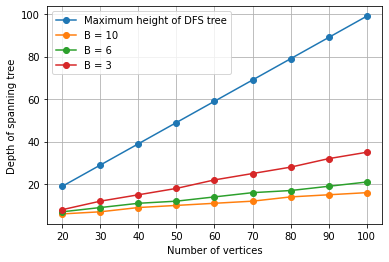}
         \caption{Complete graph}
     \end{subfigure}
    \caption{Depth of the circuit for different values of B: Erdos Renyi Graphs ($p_{edge} = 0.4,~0.6,~ 0.8$) and complete graphs}
    \label{fig:depth}
\end{figure}

We observe from the graphs in Fig~\ref{fig:depth} that the increase in the depth with $n$ for various values of $B$ is still linear. In the worst case, where the depth is $n-1$ for a graph with $n$ vertices, the slope is $\simeq 1$. In Table~\ref{tab:slope} we show the slopes of the curves for $B = 3$, $6$ and $10$ for each of the graph family considered. From the values it is evident that the slope corresponding to the increase in depth is lowered by $\simeq \frac{1}{10}$ as the value of $B$ increases.

\begin{table}[htb]
    \centering
    \caption{Variation in the slope of the increase in depth with $n$ for different values of $B$}
    \begin{tabular}{|c|c|c|c|}
    \hline
        Graph Family & $B = 3$ & $B = 6$ & $B = 10$ \\
        \hline
        Erdos Renyi ($p_{edge} = 0.4$) & 0.35 & 0.1875 & 0.1125\\
        \hline
        Erdos Renyi ($p_{edge} = 0.6$) & 0.35 & 0.1875 & 0.1125\\
        \hline
        Erdos Renyi ($p_{edge} = 0.8$) & 0.3375 & 0.1875 & 0.125\\
        \hline
        Complete graph & 0.3375 & 0.1875 & 0.125\\
        \hline
    \end{tabular}
    \label{tab:slope}
\end{table}

\subsection{Increase in the Probability of Success}

QAOA consists of executing the same circuit with the same parameters multiple times to obtain an expectation value of the cut. The performance of the algorithm is determined by this expectation value of the obtained cut. However, since our QAOA circuit, and the QAOA circuit in \cite{majumdar2021optimizing} are functionally equivalent to the traditional QAOA circuit, the performance remains unchanged. In this paper, we define success in a different way. For each iteration of the algorithm, let $\ket{\psi}$ denote the ideal state vector obtained via noiseless simulation.  As real-world quantum devices are noisy, let $\ket{\psi_e}$ denote the noisy outcome obtained via noisy simulation. We define the probability of success $P_{success} = |\braket{\psi|\psi_e}|^2$. For a graph with $n$ vertices, the optimization proposed in \cite{majumdar2021optimizing} reduced the number of CNOT gates by at most $n-1$. Since a CNOT gate is one of the most acute sources of error, the method improved $P_{success}$. However, this improvement has the overhead of increase in the depth, which exposes the circuit to more decoherence. In this paper, we have retained the $n-1$ reduction in the number of CNOT gates needed for the operator $U(H_P)$ in the ansatz and have also arrested the increase in depth to a bare minimum (refer Fig.~\ref{fig:depth} (a)-(d)). This leads to a further improvement in $P_{success}$.

\begin{figure}[htb]
     \centering
     \begin{subfigure}[b]{0.48\textwidth}
         \centering
         \includegraphics[scale=0.5]{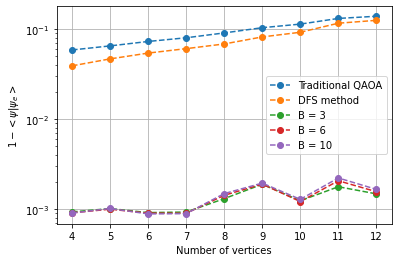}
         \caption{Erdos-Renyi graphs with $p_{edge} = 0.4$}
     \end{subfigure}
     \hfill
     \begin{subfigure}[b]{0.48\textwidth}
         \centering
         \includegraphics[scale=0.5]{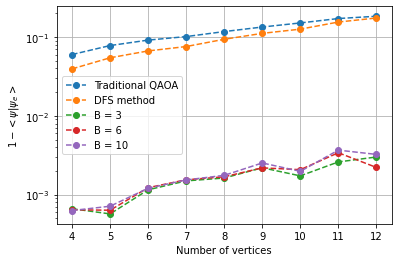}
         \caption{Erdos-Renyi graphs with $p_{edge} = 0.6$}
     \end{subfigure}
     \newline
     \begin{subfigure}[b]{0.47\textwidth}
         \centering
         \includegraphics[scale=0.5]{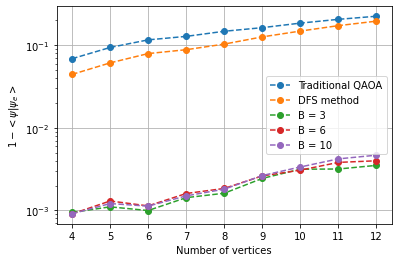}
         \caption{Erdos-Renyi graphs  with $p_{edge} = 0.8$}
     \end{subfigure}
     \hfill
     \begin{subfigure}[b]{0.47\textwidth}
         \centering
         \includegraphics[scale=0.5]{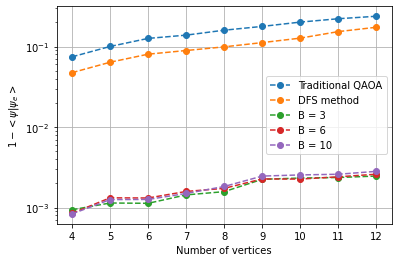}
         \caption{Complete graphs}
     \end{subfigure}
    \caption{$1-P_{success}$ for Erdos Renyi Graphs ($p_{edge} = 0.4,~0.6,~ 0.8$) and complete graphs}
    \label{fig:success}
\end{figure}

The plots corresponding to $P_{success}$ is somewhat indecipherable because the decrease in $P_{success}$ with increasing $n$ is significantly less in our method than the traditional QAOA or the optimization of \cite{majumdar2021optimizing}. Thus the changes are not easily observable in a plot for $P_{success}$. Therefore, in Fig.~\ref{fig:success} (a)-(d) we plot $(1 - P_{success}) = 1 - |\braket{\psi|\psi_e}|^2$ for the four graph families using the $ibmq\_manhattan$ noise model for the noisy simulation. We show that our proposed method decreases $(1 - P_{success})$ by more than 10 times as compared to the traditional QAOA or the optimization in \cite{majumdar2021optimizing} for Erdos-Renyi graphs with $0.4 \leq p_{edge} \leq 0.8$ and complete graph.

We observe that in Fig.~\ref{fig:success} (a)-(d), it is not evident that any betterment is achieved by increasing the value of $B$. However, this is simply because for graphs with low values of $n$, the maximum degrees are low as well. This is supported by similar results observed for low values of $n$ in Fig.~\ref{fig:depth} (a)-(d) as well. However, we see in those plots that as $n$ increases, increasing the value of $B$ leads to a better result. So we expect to see better results for $P_{success}$ as well with increasing value of $B$ as the number of vertices is increased further.

%%%%%%%%%%%%%%%%%%%%%%%%%%%%%%%%%%%%%%%%%
\section{Conclusion}
\label{sec:conclusion}

In \cite{majumdar2021optimizing} the authors showed that $n-1$ CNOT gates can be omitted from the ansatz circuit of QAOA for Max-Cut for an $n$ vertex graph if a DFS based ordering is followed while constructing the circuit. However, this led to an increase in the depth of the resulting circuit. In this paper, we have proposed a polynomial time heuristic algorithm that can find a rooted spanning tree such that the reduction in the number of CNOT gates is retained while significantly arresting the increase in depth. Our method is able to reduce the increase in depth by $\simeq \frac{1}{10}$ as compared to the circuit in \cite{majumdar2021optimizing}. This, in its turn, leads to a significant increase in the success probability of the algorithm, since (i) the reduction in CNOT gates is retained, and (ii) the increase in depth is lowered thus making the circuit less susceptible to relaxation error. Our proposed heuristic for circuit synthesis is thus expected to provide a novel scheme for mitigating the effect of error in QAOA for Max-Cut, and can be used together with other error mitigation scheme.

Our results show that higher branching factor is better to arrest the increase in depth as the number of vertices increases. A future prospect can be to study an approximate value of the branching factor for a particular $n$ such that the increase in depth is minimum.

%%%%%%%%%%%%%%%%%%%%%%%%%%%%%%%%%%%%%%%%%
\section*{Acknowledgement}

We acknowledge the use of IBM Quantum services for this work. The views expressed are those of the authors, and do not reflect the official policy or position of IBM or the IBM Quantum team. In this paper we have used the noise model of \textit{ibmq\_manhattan}, which is one of IBM Quantum Hummingbird r2 Processors. 

\bibliographystyle{unsrt}
\bibliography{main}
	
\end{document}